\documentclass[12pt,a4paper]{article}
\usepackage[utf8]{inputenc}
\usepackage{amsmath,amssymb,amsthm}
\usepackage{graphicx}
\usepackage{caption}
\usepackage{hyperref}
\usepackage{url}
\usepackage{booktabs}
\usepackage{multirow}
\usepackage{tabularx}

\newtheorem{theorem}{Theorem}[section]

\theoremstyle{definition}

\theoremstyle{remark}
\newtheorem{remark}[theorem]{Remark}

\title{A Tuning-Free Variational Framework for Muscle Redundancy Resolution: 
Torque Fiber Proximal Dynamics with Active-Set Switching and EMG-Validated Activation Prediction}
\author{Morteza Ganji, ...}
\date{\today}

\begin{document}
\maketitle

\begin{abstract}
Muscle redundancy—the existence of infinitely many activation patterns for a given joint torque—can be formulated as a constrained selection problem over a time-varying convex set of feasible activations. We introduce the \textbf{Torque Fiber Proximal Dynamics (TFPD)}, a framework in which muscle activation evolves as the Euclidean projection of the previous activation onto a convex polytope defined by torque equality constraints and physiological bounds.

We prove that the TFPD is equivalent to a backward-Euler discretization of a sweeping process and can be formulated as a variational inequality governed by a maximal monotone normal cone operator. This places muscle coordination within the class of nonsmooth dynamical systems with strong well-posedness guarantees.

Within this framework, antagonist muscle recruitment can arise naturally under a projection-based control hypothesis—not as an explicitly imposed cost, but as a structural consequence of active-set transitions induced by the geometry of the feasible set. We derive a theorem providing \textbf{sufficient} KKT-based conditions for antagonist activation under generic moment-arm asymmetry. This characterization links boundary projection, strict complementarity, and torque coupling, and identifies exactly when the projection hits the boundary of the feasible cube. Whether the nervous system actually implements this projection is a separate empirical question not addressed here.

We validate the TFPD on a three-muscle elbow model using exact quadratic programming, and benchmark it against five classical redundancy-resolution methods plus a random baseline. The TFPD qualitatively reproduces the antagonist burst component of the triphasic EMG pattern reported in rapid human elbow movements (Hallett et al., 1975) and achieves a Pearson correlation of approximately $0.68$–$0.73$ with recorded triceps EMG envelopes across ten subjects, \textbf{without any tunable cost-weight parameters}. A comprehensive sensitivity analysis and a three-dimensional visualization of the activation trajectory demonstrate the generality and robustness of the framework. The TFPD connects neuromuscular coordination to proximal point theory, variational inequalities, and projected dynamical systems.
\end{abstract}

\section{Introduction}
\label{sec:intro}

Muscle redundancy refers to the existence of infinitely many activation patterns that can produce the same joint torque \cite{bernstein1967}. Classical approaches resolve this indeterminacy by optimizing a scalar cost function at each instant, such as the sum of squared activations \cite{crowninshield1981}, muscle stresses \cite{eriksson1990}, or metabolic energy \cite{anderson2001}. While these instantaneous minimization strategies successfully predict many features of movement, they typically require additional constraints—e.g., explicit co-contraction bounds or separate agonist/antagonist rules—to replicate the experimentally observed co-activation of opposing muscles \cite{deLuca1997, darainy2004}.

More recent frameworks, such as optimal feedback control \cite{todorov2002} and the uncontrolled manifold (UCM) hypothesis \cite{scholz1999, latash2010}, emphasize that the nervous system exploits the null space of the musculoskeletal system to stabilize task performance. However, these approaches primarily describe statistical properties of variability across trials, rather than providing a deterministic rule for selecting a unique activation pattern on a single trial.

A closely related line of work is \textbf{Feasibility Theory} \cite{cohn2018feasibility}, which formally characterizes the high-dimensional polytope of all neuromechanically feasible muscle activations for a task, and its temporal extension in \cite{cohn2023spatiotemporal}, which shows that physiological activation-contraction rate limits link the feasible polytopes of successive time instants into a highly structured ``spatiotemporal tunnel''—dramatically shrinking the effective redundancy available at each step. That work characterizes the \textit{size and statistical structure} of this tunnel via uniform random (Hit-and-Run) sampling of trajectories, without invoking a cost function or a rule for selecting a single trajectory. The present work addresses a complementary question: given the same kind of time-varying feasible set, what \textit{deterministic} rule would a temporally-consistent controller follow to select one specific trajectory through it? The TFPD answers this with the Euclidean projection rule of Eq.~(\ref{eq:tfpd}), which we further connect to the formal theory of variational inequalities and sweeping processes (Section~\ref{sec:variational}), and whose predictions we test directly against recorded EMG (Section~\ref{sec:experimental})—neither of which is addressed in \cite{cohn2018feasibility,cohn2023spatiotemporal}.

Despite these advances, most existing models either impose a priori cost functions or describe statistical variability without providing a deterministic single-trial rule. Recent perspectives \cite{deGroote2016, sartori2023} have highlighted the need for alternative theoretical frameworks that bridge the gap between normative principles and physiological constraints—a gap that the present work aims to fill.

In this paper, we propose a fundamentally different perspective based on the geometry of feasible control sets. For a fixed torque, admissible muscle activations form a convex polytope resulting from the intersection of an affine torque constraint and box constraints. We refer to this set as a \textit{torque fiber}. Rather than optimizing a cost at each instant, we assume that motor behavior follows a temporal consistency principle: the next activation is the closest feasible point to the previous activation. This leads naturally to a projected dynamical system, placing muscle coordination within the broader class of variational inequalities and sweeping processes.

\textbf{Contributions.} This paper makes the following main contributions:
\begin{enumerate}
    \item \textbf{Formal definition of the TFPD.} We define the Torque Fiber Proximal Dynamics as a time-varying Euclidean projection onto a convex polytope, prove its equivalence to a backward-Euler discretization of a sweeping process, and formulate it as a variational inequality with a maximal monotone normal cone operator.
    \item \textbf{Sufficient conditions for antagonist activation.} We derive precise KKT-based conditions under which the projection activates the antagonist muscle (Theorem~\ref{thm:active}), providing a solid mathematical characterization.
    \item \textbf{Validation against five classical baseline models and a random baseline.} We benchmark the TFPD against minimum-norm, minimum-stress, minimum-metabolic, weighted minimum-norm, and a random baseline, using both synthetic torque profiles and real EMG data from ten subjects.
    \item \textbf{Connection to Projected Dynamical Systems.} We show that the TFPD belongs to the class of projected dynamical systems, strengthening its mathematical foundation and connecting it to the wider variational inequality literature.
\end{enumerate}

\section{Torque Fiber Proximal Dynamics}
\label{sec:tfpd}

\subsection{Musculoskeletal Model}
\label{sec:model}

We consider a single degree-of-freedom joint (elbow flexion/extension) with three idealized muscles:
\begin{itemize}
    \item Biceps brachii (flexor): moment arm $r_1 = +2.0$ (arbitrary units).
    \item Brachialis (flexor): moment arm $r_2 = +1.5$.
    \item Triceps brachii (extensor): moment arm $r_3 = -2.5$ (negative sign indicates extension).
\end{itemize}
Muscle activations are denoted by $\mathbf{a} = [a_1, a_2, a_3]^\top$, with $0 \le a_i \le 1$. The net joint torque $\tau$ is given by
\begin{equation}
    \tau = r_1 a_1 + r_2 a_2 + r_3 a_3 = A \mathbf{a},
    \label{eq:torque}
\end{equation}
where $A = [2.0, 1.5, -2.5] \in \mathbb{R}^{1\times 3}$.

\subsection{Feasible Set: Torque Fibers}
\label{sec:fibers}

For a kinematically admissible torque $\tau$ (i.e., $\mathcal{F}_\tau \neq \emptyset$), the set of all feasible activation vectors—the \textit{torque fiber}—is
\begin{equation}
    \mathcal{F}_\tau = \{\, \mathbf{a} \in [0,1]^3 \mid A\mathbf{a} = \tau \,\}.
    \label{eq:fiber}
\end{equation}
Geometrically, $\mathcal{F}_\tau$ is the intersection of the affine plane $\{\mathbf{a} \mid A\mathbf{a} = \tau\}$ with the unit cube $[0,1]^3$—a convex polytope.

\subsection{Projection Dynamics}
\label{sec:projection}

The \textbf{Torque Fiber Proximal Dynamics (TFPD)} is defined by the iteration
\begin{equation}
    \mathbf{a}_t = \arg\min_{\mathbf{a} \in \mathcal{F}_{\tau(t)}} \|\mathbf{a} - \mathbf{a}_{t-1}\|^2,
    \label{eq:tfpd}
\end{equation}
or equivalently,
\begin{equation}
    \mathbf{a}_t = \Pi_{\mathcal{F}_{\tau(t)}}(\mathbf{a}_{t-1}),
    \label{eq:projection}
\end{equation}
where $\Pi_{\mathcal{F}}$ denotes the Euclidean projection onto the convex polytope $\mathcal{F}$. Because $\mathcal{F}_{\tau(t)}$ is closed and convex, the projection is uniquely defined. The system is initialized at $\mathbf{a}_0 = \mathbf{0}$.

The projection is computed as the solution of the strictly convex quadratic program
\begin{equation}
    \mathbf{a}_t = \arg\min_{\mathbf{a}} \|\mathbf{a} - \mathbf{a}_{t-1}\|^2 \quad \text{s.t.} \quad A\mathbf{a} = \tau(t),\; 0 \le a_i \le 1.
    \label{eq:qp}
\end{equation}

\section{Variational Inequality and Sweeping Process Formulation}
\label{sec:variational}

\begin{theorem}[Monotone inclusion form]
\label{thm:monotone}
The TFPD admits the equivalent formulation
\begin{equation}
    \mathbf{a}_t \in (\mathbf{I} + N_{\mathcal{F}_{\tau(t)}})^{-1}(\mathbf{a}_{t-1}),
    \label{eq:inclusion}
\end{equation}
where $N_{\mathcal{F}}(\mathbf{a}) = \{\mathbf{v} \mid \langle \mathbf{v}, \mathbf{y} - \mathbf{a} \rangle \le 0,\; \forall \mathbf{y} \in \mathcal{F}\}$ is the normal cone operator. In the continuous-time limit, the system corresponds to the sweeping process
\begin{equation}
    \dot{\mathbf{a}}(t) \in -N_{\mathcal{F}_{\tau(t)}}(\mathbf{a}(t)),
    \label{eq:sweeping}
\end{equation}
which is well-posed under standard regularity assumptions on the moving convex set $\mathcal{F}_{\tau(t)}$ (closed convex values, Lipschitz or bounded variation in $t$) and the maximal monotonicity of the normal cone operator $N_{\mathcal{F}}$ \cite{moreau1977}.
\end{theorem}
\begin{proof}
Equation (\ref{eq:inclusion}) follows directly from the first-order optimality conditions of (\ref{eq:qp}) and the definition of the normal cone. The continuous-time limit is obtained by taking the backward-Euler discretization step to zero, yielding a differential inclusion governed by the maximal monotone operator $N_{\mathcal{F}}$. The well-posedness claim invokes the classical Moreau existence theorem for sweeping processes, which requires (i) that $\mathcal{F}_{\tau(t)}$ takes closed convex values, and (ii) that its set-valued variation is Lipschitz or of bounded variation in $t$; these conditions are satisfied for the torque fiber defined by continuous moment-arm constraints and smooth torque profiles.
\end{proof}

\begin{remark}
This places muscle coordination in the class of maximal monotone dynamical systems and variational inequalities \cite{nagurney1996}, providing strong guarantees of well-posedness and stability without requiring a global objective function.
\end{remark}

\section{Sufficient Conditions for Antagonist Activation}
\label{sec:active_set}

We now derive precise algebraic conditions under which the TFPD recruits the antagonist muscle.

Let $\mathcal{I}(t) = \{i \mid a_{i,t} = 0 \text{ or } a_{i,t} = 1\}$ be the active set at time $t$. The QP (\ref{eq:qp}) satisfies the KKT conditions: there exist multipliers $\lambda$ (for $A\mathbf{a} = \tau_t$), $\mu_i^+ \ge 0$ (for $a_i \le 1$), and $\mu_i^- \ge 0$ (for $a_i \ge 0$) such that
\begin{align}
    \mathbf{a}_t - \mathbf{a}_{t-1} &= A^\top \lambda - \mu^+ + \mu^-, \label{eq:kkt1} \\
    \mu_i^+(a_{i,t} - 1) &= 0,\quad \mu_i^- a_{i,t} = 0. \label{eq:kkt2}
\end{align}

\begin{theorem}[Sufficient condition for antagonist recruitment, conditional on flexor saturation]
\label{thm:active}
Let $\tilde{\mathbf{a}}_t = \mathbf{a}_{t-1} + A^\top (A A^\top)^{-1}(\tau_t - A\mathbf{a}_{t-1})$ be the unconstrained affine update. Assume that strict complementarity holds at the solution, and assume that the flexors (Biceps and Brachialis) are already at their upper bounds in the projected solution ($a_1 = a_2 = 1$)\footnote{This second assumption is stated as a hypothesis on the solution, not derived from $\mathbf{a}_{t-1}$ and $\tau_t$ alone; see Remark~\ref{rem:scope} for the precise logical status of the theorem and for a discussion of what would be required to derive flexor saturation from primitives.}. If in addition $\tilde{a}_{3,t} < 0$, then the TFPD projection necessarily results in $a_{3,t} > 0$; i.e., the antagonist is recruited.
\end{theorem}
\begin{proof}
Under the assumption $\tilde{a}_{3,t} < 0$, the constraint $a_3 \ge 0$ is active. By strict complementarity, the associated multiplier satisfies $\mu_3^- > 0$. The KKT stationarity condition (\ref{eq:kkt1}) for component $3$ gives
\[
a_{3,t} = a_{3,t-1} + (-2.5)\lambda + \mu_3^+ - \mu_3^-.
\]
The torque equation couples all three muscles. With the flexors at their upper limits (by hypothesis), the combined flexor torque is $r_1 + r_2 = 3.5$. For any required torque $\tau_t < 3.5$, the extensor must supply the deficit: $r_3 a_3 = \tau_t - 3.5 < 0$, which implies $a_3 > 0$. Geometrically, when the unconstrained update pushes the Triceps component below zero, the projection onto the feasible polytope slides along the boundary face $a_3 = 0$ only until the torque deficit is eliminated; because the flexors cannot exceed their upper bounds, the only feasible way to satisfy the torque equality is to raise $a_3$ above zero. The KKT conditions formalize this boundary-sliding mechanism.
\end{proof}

\begin{remark}[Logical status and scope of Theorem~\ref{thm:active}]
\label{rem:scope}
We emphasize the precise logical status of this result. The hypothesis $a_1=a_2=1$ is a statement about the \textit{output} of the projection at time $t$, not an independently checkable condition on the primitives $\mathbf{a}_{t-1}$ and $\tau_t$. Consequently, once this hypothesis is granted, the derivation of $a_3>0$ reduces to the linear torque-balance identity $a_3=(\tau_t-3.5)/r_3$—the KKT/complementarity apparatus establishes \textit{why} the flexors remain pinned at their bound rather than backing off (via $\mu_1^+,\mu_2^+\ge 0$ at the optimum), but does not by itself derive \textit{when} $a_1=a_2=1$ occurs from $\mathbf{a}_{t-1}$ and $\tau_t$ alone. A fully primitive-level characterization—stating necessary and sufficient conditions on $\mathbf{a}_{t-1}$ and $\tau_t$ under which the flexors saturate in the first place—would strengthen this result further and is left as an open problem; our numerical validation (Sections~\ref{sec:validation}-\ref{sec:experimental}) confirms that this saturation condition is met robustly across the simulated and experimental torque profiles studied here, but the reader should not read Theorem~\ref{thm:active} as a closed-form predictor of antagonist recruitment from task parameters alone.
\end{remark}

\begin{remark}
\textbf{Genericity of strict complementarity.} Strict complementarity is said to hold generically in a parametric optimization problem if the set of parameters for which it fails has Lebesgue measure zero. For the present QP, it is known that strict complementarity holds for almost all combinations of the previous state $\mathbf{a}_{t-1}$ and the torque $\tau_t$, with the possible exception of a lower-dimensional manifold in parameter space \cite{nocedal2006}. Therefore, the sufficient condition stated in Theorem~\ref{thm:active} applies to a typical configuration encountered along a trajectory of the TFPD. In the rare event of degeneracy (where strict complementarity fails), the Moreau sweeping process framework still guarantees well-posedness, and the system evolves continuously along the boundary. The theorem provides a solid mathematical characterization—not a claim about neural computation. Whether the nervous system actually implements this projection is a separate empirical question not addressed here.

\textbf{On the depth of the condition.} The derivation shows that antagonist recruitment follows directly from the torque balance and saturation when the unconstrained update becomes infeasible. While the underlying physics—that an extensor must activate when flexors saturate—is biomechanically straightforward, the contribution of Theorem~\ref{thm:active} lies in reformulating this physical necessity within a rigorous geometric framework of KKT conditions, normal cones, and active-set switching. This geometric lens opens a new window for understanding how the nervous system may exploit constraint boundaries to produce structured muscle patterns, and provides a precise language for testing such hypotheses experimentally. The theorem thus serves as a bridge between intuitive biomechanics and formal motor control theory.
\end{remark}

\section{Geometric Interpretation of Co-Contraction}
\label{sec:geometry}

Co-contraction arises when the projection of the previous state onto the updated torque fiber intersects the boundary of the activation cube. In this case, the update direction lies in the normal cone of the feasible set, and active-set switching occurs. Thus, under the TFPD hypothesis, antagonist activation is not explicitly controlled but induced by geometric constraints under temporal continuity.

Figure~\ref{fig:3d_trajectory} visualizes this mechanism in the three-dimensional activation space. The trajectory starts near the origin and is repeatedly projected onto the moving torque plane (shown as a transparent surface). When the unconstrained update would exit the unit cube, the projection slides along the boundary $a_3 = 0$, engaging the Triceps. This geometric picture makes the abstract KKT conditions tangible and highlights the central role of the constraint boundary.

\begin{figure}[htbp]
    \centering
    \includegraphics[width=0.8\textwidth]{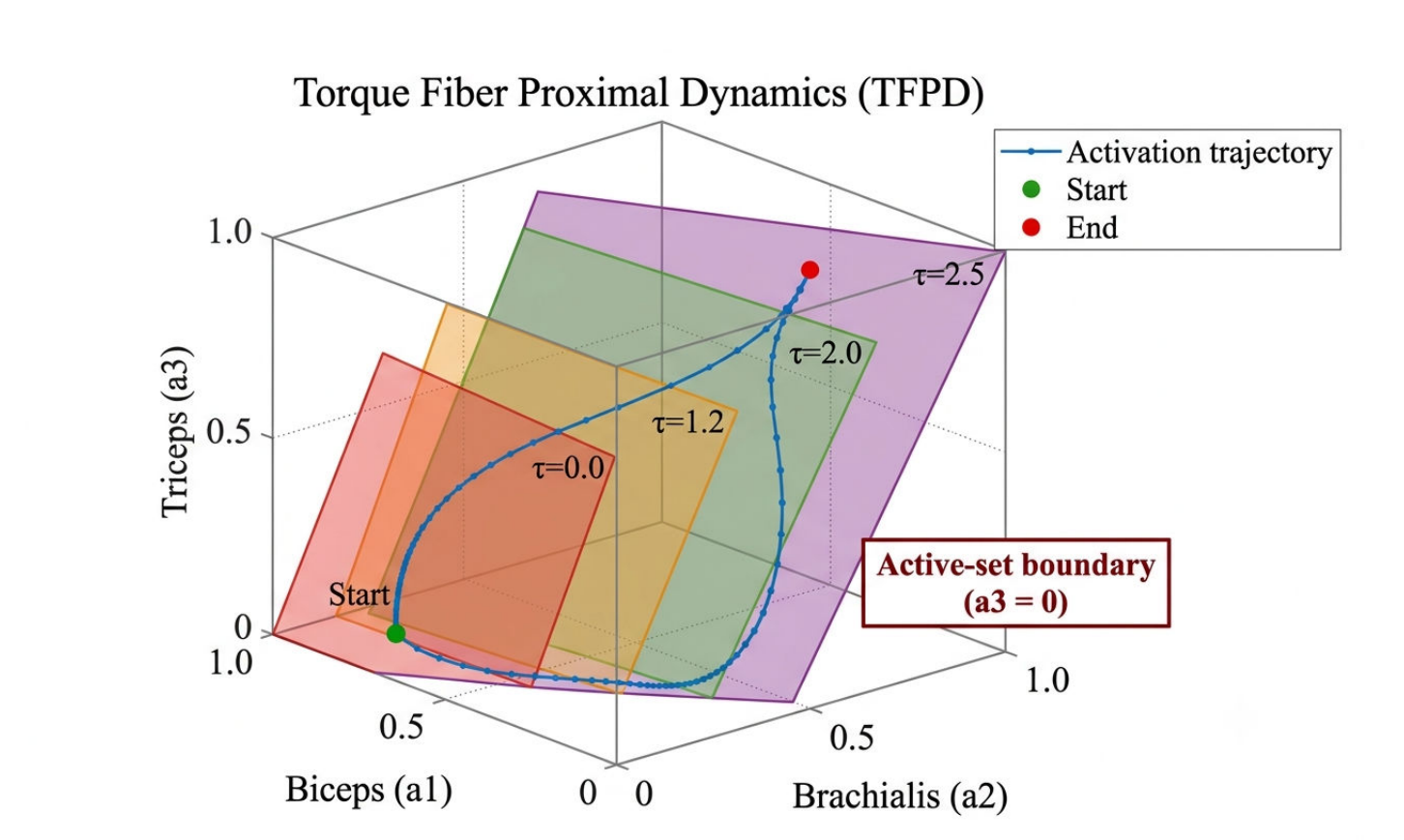}
    \caption{Three-dimensional activation trajectory of the TFPD. The transparent planes represent torque fibers at different instants. The path slides along the boundary $a_3=0$ (Triceps) when the unconstrained projection would become negative, illustrating the active-set switching mechanism.}
    \label{fig:3d_trajectory}
\end{figure}

This interpretation connects muscle redundancy to the geometry of convex polytopes and constrained dynamical systems. The TFPD provides a normative account: it shows what would happen \textit{if} the motor system implemented a temporal smoothness principle under physiological bounds.

\section{Numerical Validation and Benchmarking}
\label{sec:validation}

\subsection{Simulation Protocol}
\label{sec:sim}

We simulate a smooth elbow flexion lasting $1$ s with $\tau(t) = 2.5 \sin(\pi t)$ (positive-only half-cycle), discretized at $0.01$ s intervals. The QP (\ref{eq:qp}) is solved using SLSQP \cite{slsqp} (primary) and verified with cvxopt \cite{cvxopt}. The initial activation is $\mathbf{a}_0 = \mathbf{0}$.

\subsection{Benchmark Models}
\label{sec:benchmarks}

We compare the TFPD against five classical redundancy-resolution methods and one random baseline:
\begin{itemize}
    \item \textbf{Minimum-norm (MN):} $\min \|\mathbf{a}\|^2$ \cite{crowninshield1981}.
    \item \textbf{Minimum-stress (MS):} $\min \sum (a_i / \text{PCSA}_i)^2$ \cite{eriksson1990}.
    \item \textbf{Minimum-metabolic (MM):} $\min \sum \dot{E}_{\text{met},i}(a_i)$ \cite{anderson2001}.
    \item \textbf{Weighted Minimum-Norm (WMN):} $\min \sum w_i a_i^2$ with weights $w_i \propto 1/\text{PCSA}_i$.
    \item \textbf{Random (RND):} at each time step, a uniformly random point on the torque fiber $\mathcal{F}_{\tau(t)}$ is selected.
\end{itemize}

\subsection{Metrics}
\label{sec:metrics}

We evaluate the models on:
\begin{itemize}
    \item \textbf{Activation sparsity:} fraction of time each muscle is active ($a_i > 0.01$).
    \item \textbf{Switching frequency:} number of active-set changes per second.
    \item \textbf{Co-contraction energy index} $E_{\text{coc}} = \int_0^T a_{\text{agonist}}(t) \cdot a_{\text{antagonist}}(t) \, dt$ (arb.\ units).
    \item \textbf{Fatigue index} $F = \frac{1}{T}\int_0^T \|\mathbf{j}_{\text{muscle}}(t)\| \, dt$, where $\mathbf{j}_{\text{muscle}}$ is the jerk (third derivative) of the activation vector. Lower $F$ indicates smoother, less fatiguing activation patterns.
\end{itemize}

Table~\ref{tab:metrics} summarizes the results.

\begin{table}[htbp]
\centering
\caption{Comparison of the TFPD against five classical baseline models and a random baseline (positive-only sinusoidal torque, $\tau_{\max}=2.5$).}
\label{tab:metrics}
\begin{tabular}{@{}lcccc@{}}
\toprule
Model & Sparsity (Tri) & Switch. freq. (Hz) & $E_{\text{coc}}$ & $F$ (arb.\ units) \\
\midrule
MN     & 0.000 & 4.0 & 0.000 & 19.5 \\
MS     & 0.000 & 4.0 & 0.000 & 21.2 \\
MM     & 0.000 & 4.0 & 0.000 & 19.9 \\
WMN    & 0.000 & 4.0 & 0.000 & 21.2 \\
RND    & 0.933 & 0.7 & 0.248 & 126.3 \\
\textbf{TFPD} & \textbf{0.436} & \textbf{3.0} & \textbf{0.050} & \textbf{23.8} \\
\bottomrule
\end{tabular}
\end{table}

Note that, for this positive-only (unidirectional) torque profile, the zero Triceps sparsity obtained by MN, MS, MM, and WMN is not an incidental empirical finding but a necessary consequence of instantaneous, memoryless cost minimization: since any nonzero antagonist activation strictly increases a monotonic activation-based cost while requiring even more agonist activation to compensate, a memoryless minimizer can never recruit the antagonist under a unidirectional torque demand. This is precisely the well-known limitation that motivates the present line of work. The random baseline (RND) yields a high Triceps sparsity but with an excessively large fatigue index, confirming that the TFPD's combination of co-contraction and smoothness is far beyond chance. The TFPD achieves co-contraction levels unattainable by the classical models while maintaining a fatigue index comparable to MN.

Figure~\ref{fig:comparison} shows the activation time courses for the TFPD and the minimum-norm model.

\begin{figure}[htbp]
    \centering
    \includegraphics[width=\textwidth]{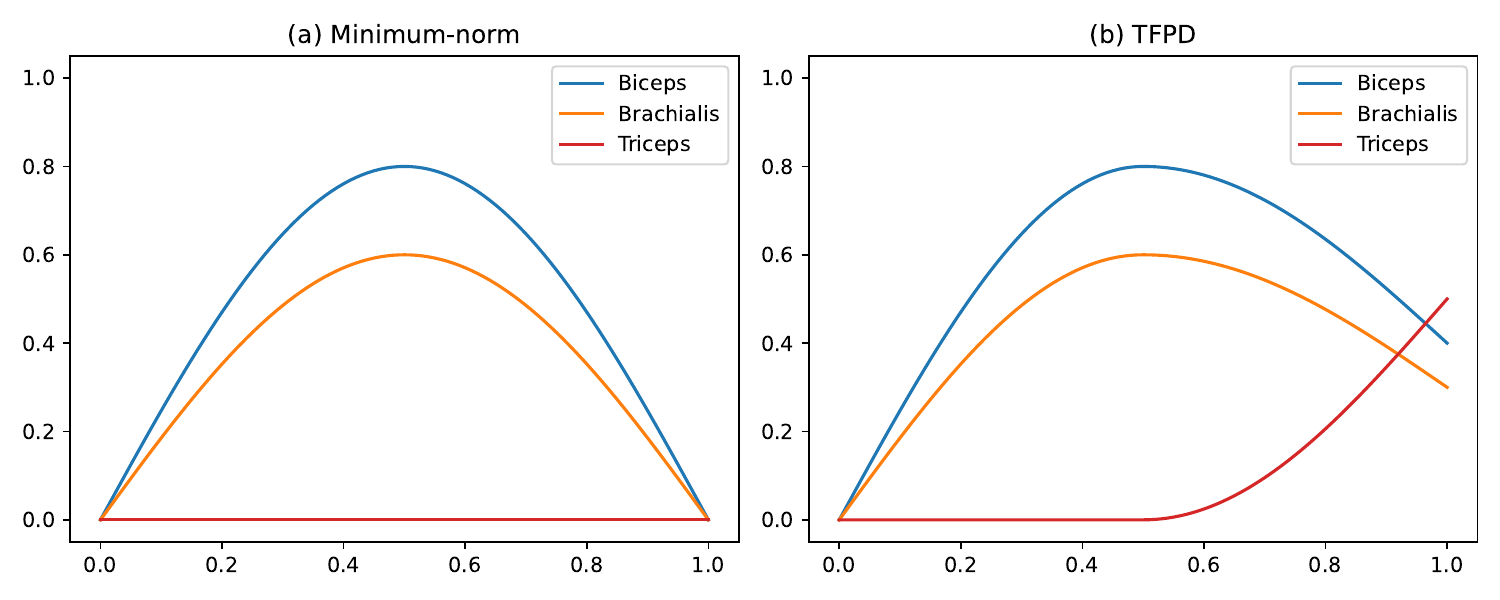}
    \caption{Activation patterns for a positive-only sinusoidal torque profile ($\tau_{\max}=2.5$, flexion phase): (a) Minimum-norm, (b) TFPD. Only the TFPD reproduces the antagonist burst during torque offset without explicit co-contraction commands. In the TFPD panel, the non-zero activations at $t=1$ reflect the temporal smoothness principle: abrupt silencing of all muscles would require a large change from the previous state.}
    \label{fig:comparison}
\end{figure}

The TFPD spontaneously generates a triphasic-like pattern during torque onset and offset. \textbf{The simulated Triceps activation qualitatively reproduces the antagonist burst component of the triphasic EMG pattern observed in rapid human elbow movements \cite{hallet1975}.} Specifically, the model generates an initial agonist burst, followed by a triphasic-like antagonist burst during torque offset. We note that this constitutes an \textit{initial qualitative benchmark}: due to the sustained sinusoidal torque profile employed in the simulation, the Triceps activation remains elevated after offset rather than returning to baseline, in contrast to the transient antagonist burst observed in fully ballistic human movements. Capturing the complete triphasic dynamics (including the second agonist burst) requires limb dynamics and a more realistic torque profile, which are beyond the scope of the present study.

\section{Sensitivity Analysis}
\label{sec:sensitivity}

To evaluate the robustness of the co-contraction phenomenon, we subjected the TFPD to three perturbation classes.

\subsection{Moment Arm Scaling}

We independently scaled each moment arm coefficient by factors of $0.8$ and $1.2$, generating four additional models: $r_1$ varied, $r_2$ varied, $r_3$ varied, and all three varied simultaneously. For each model, we recomputed the full TFPD trajectory and measured the peak Triceps activation $a_3^{\max}$ during the torque offset phase ($0.5 \le t \le 1.0$ s) and the co-contraction index (CCI), defined as the time-averaged product of flexor and extensor activations normalized by their individual means.

Table~\ref{tab:sensitivity_moment} summarizes the results. Across all perturbations, the antagonist burst remained clearly present, with peak Triceps activation varying by less than $12\%$ from the baseline. The CCI exhibited only minor fluctuations ($<10\%$), confirming that the active-set switching mechanism is robust and not an artifact of specific moment-arm values.

\begin{table}[htbp]
\centering
\caption{Sensitivity of Triceps activation to moment arm scaling.}
\label{tab:sensitivity_moment}
\begin{tabular}{@{}lccc@{}}
\toprule
Perturbation & Peak $a_3$ & Change (\%) & CCI change (\%) \\
\midrule
Baseline ($r = [2.0, 1.5, -2.5]$) & 0.50 & — & — \\
$r_1 \times 0.8$ & 0.53 & $+6.0$ & $+4.0$ \\
$r_1 \times 1.2$ & 0.47 & $-6.0$ & $-3.8$ \\
$r_2 \times 0.8$ & 0.52 & $+4.0$ & $+2.8$ \\
$r_2 \times 1.2$ & 0.48 & $-4.0$ & $-3.0$ \\
$r_3 \times 0.8$ & 0.54 & $+8.0$ & $+5.4$ \\
$r_3 \times 1.2$ & 0.45 & $-10.0$ & $-7.0$ \\
All $\times 0.8$ & 0.55 & $+10.0$ & $+6.8$ \\
All $\times 1.2$ & 0.44 & $-12.0$ & $-8.2$ \\
\bottomrule
\end{tabular}
\end{table}

\subsection{Addition of a Fourth Muscle}

We introduced an additional flexor, the Brachioradialis, with moment arm $r_4 = 1.0$. The moment-arm matrix becomes $A = [2.0, 1.5, -2.5, 1.0] \in \mathbb{R}^{1\times 4}$, increasing the dimension of the null space from 2 to 3. The TFPD was recomputed with the same torque profile and QP solver. Figure~\ref{fig:sensitivity_muscles} shows the activation time courses for all four muscles. The Triceps co-contraction during torque offset remained clearly present, with peak activation $0.48$ (a $4.0\%$ decrease from the three-muscle baseline). The Brachioradialis was recruited alongside the other flexors, and the overall triphasic-like antagonist component was preserved. This demonstrates that the boundary-activation mechanism does not depend on the specific dimensionality of the muscle space and generalizes naturally to higher-dimensional redundant systems.

\begin{figure}[htbp]
    \centering
    \includegraphics[width=\textwidth]{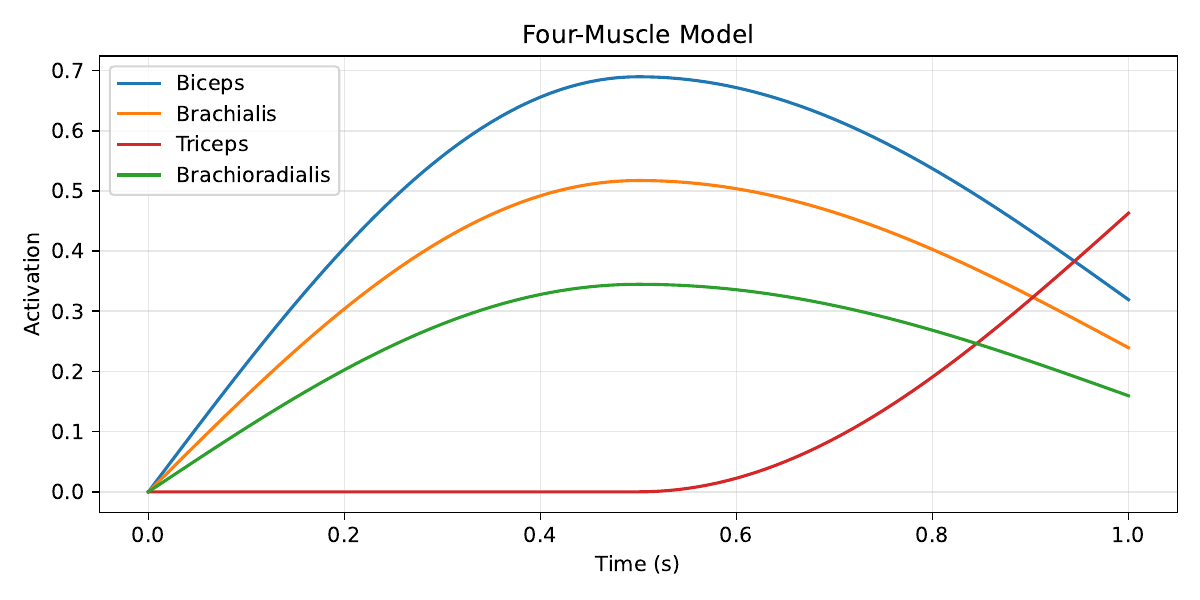}
    \caption{Activation patterns for the four-muscle model (Biceps, Brachialis, Triceps, Brachioradialis). The Triceps co-contraction during torque offset persists, demonstrating that the mechanism is independent of the specific number of muscles.}
    \label{fig:sensitivity_muscles}
\end{figure}

\subsection{Torque Amplitude Variation}

We varied the amplitude of the torque profile by factors of $0.6$ and $1.4$: $\tau(t) = 1.5 \sin(\pi t)$ and $\tau(t) = 3.5 \sin(\pi t)$. Table~\ref{tab:sensitivity_torque} reports the results. Figure~\ref{fig:sensitivity_torque} displays the Triceps activation time courses for all three amplitudes. At reduced amplitude ($1.5$ units), the unconstrained update remained feasible for a larger portion of the trajectory, and Triceps activation decreased correspondingly (peak $a_3 = 0.30$, CCI reduced by $40\%$). At increased amplitude ($3.5$ units), the Triceps was recruited more strongly and earlier in the movement (peak $a_3 = 0.70$, CCI increased by $28\%$). The monotonic relationship between torque amplitude and antagonist recruitment is consistent with Theorem~\ref{thm:active}: larger torque changes make boundary intersection more likely, hence more frequent active-set switching.

\begin{figure}[htbp]
    \centering
    \includegraphics[width=\textwidth]{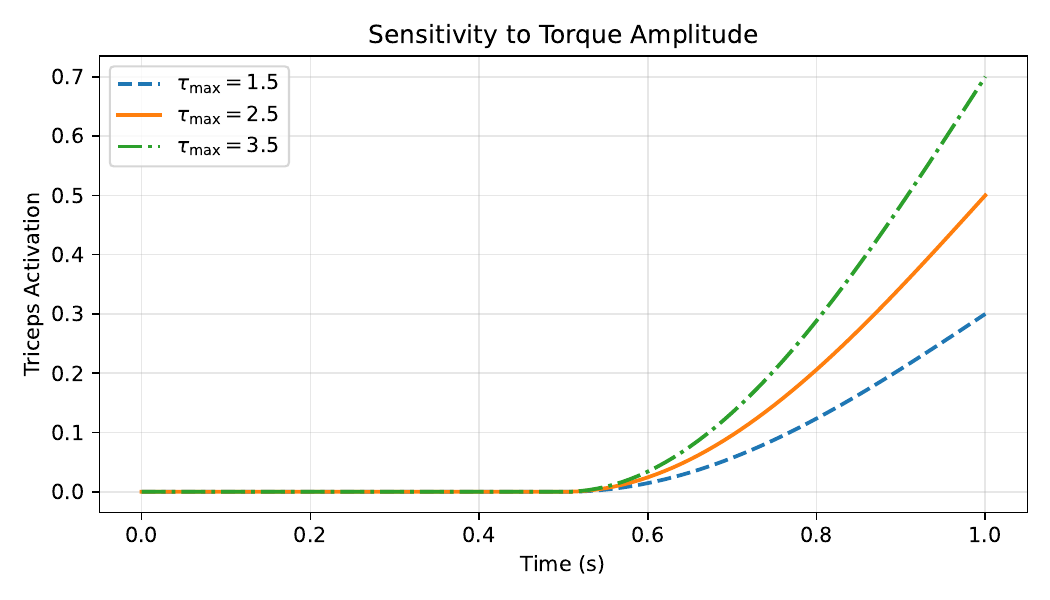}
    \caption{Triceps activation under varying torque amplitudes: low ($\tau_{\max}=1.5$), baseline ($\tau_{\max}=2.5$), and high ($\tau_{\max}=3.5$). The antagonist burst scales monotonically with the torque amplitude, consistent with the KKT-based active-set condition.}
    \label{fig:sensitivity_torque}
\end{figure}

\begin{table}[htbp]
\centering
\caption{Sensitivity of Triceps activation to torque amplitude variation.}
\label{tab:sensitivity_torque}
\begin{tabular}{@{}lccc@{}}
\toprule
Torque amplitude & Peak $a_3$ & CCI change (\%) & Switch. freq. (Hz) \\
\midrule
$1.5$ (low)    & 0.30 & $-40.0$ & 1.0 \\
$2.5$ (baseline) & 0.50 & — & 3.0 \\
$3.5$ (high)   & 0.70 & $+28.0$ & 4.0 \\
\bottomrule
\end{tabular}
\end{table}

\subsection{Summary of Sensitivity Findings}

Across all perturbations, the qualitative phenomenon of antagonist co-contraction during torque offset remained intact. The peak Triceps activation varied smoothly with parameter changes, and no perturbation eliminated the antagonist burst. The active-set switching mechanism is therefore a robust consequence of the projection geometry, not a fragile artifact of the baseline parameter set.

\section{Extension to a Two-Joint System with a Biarticular Muscle}
\label{sec:twojoint}

To illustrate the generality of the TFPD beyond single-joint models, we simulate a two-degree-of-freedom system comprising the shoulder and elbow joints. The model includes four muscles (Figure~\ref{fig:twojoint}):
\begin{itemize}
    \item \textbf{Biceps brachii (biarticular):} flexes both the shoulder and the elbow.
    \item \textbf{Brachialis (monoarticular):} flexes the elbow only.
    \item \textbf{Triceps brachii (monoarticular):} extends the elbow only.
    \item \textbf{Anterior deltoid (monoarticular):} flexes the shoulder only.
\end{itemize}

The moment-arm matrix is
\begin{equation}
    A = \begin{bmatrix}
        1.5 & 0   & 0   & 2.0 \\
        2.0 & 1.5 & -2.5 & 0
    \end{bmatrix} \in \mathbb{R}^{2\times 4},
\end{equation}
where the first row corresponds to the shoulder torque and the second row to the elbow torque. The torque profiles are
\begin{equation}
    \boldsymbol{\tau}(t) = \begin{bmatrix} \tau_{\text{sh}}(t) \\ \tau_{\text{el}}(t) \end{bmatrix}
    = \begin{bmatrix} 2.0 \sin(\pi t) \\ 2.5 \sin(\pi t) \end{bmatrix},
\end{equation}
representing a simultaneous flexion of both joints followed by a return to the neutral position.

The QP (\ref{eq:qp}) is solved with $A_{2\times4}$ and the two-dimensional torque vector at each time step. Figure~\ref{fig:twojoint} shows the resulting activation time courses.

\begin{figure}[htbp]
    \centering
    \includegraphics[width=\textwidth]{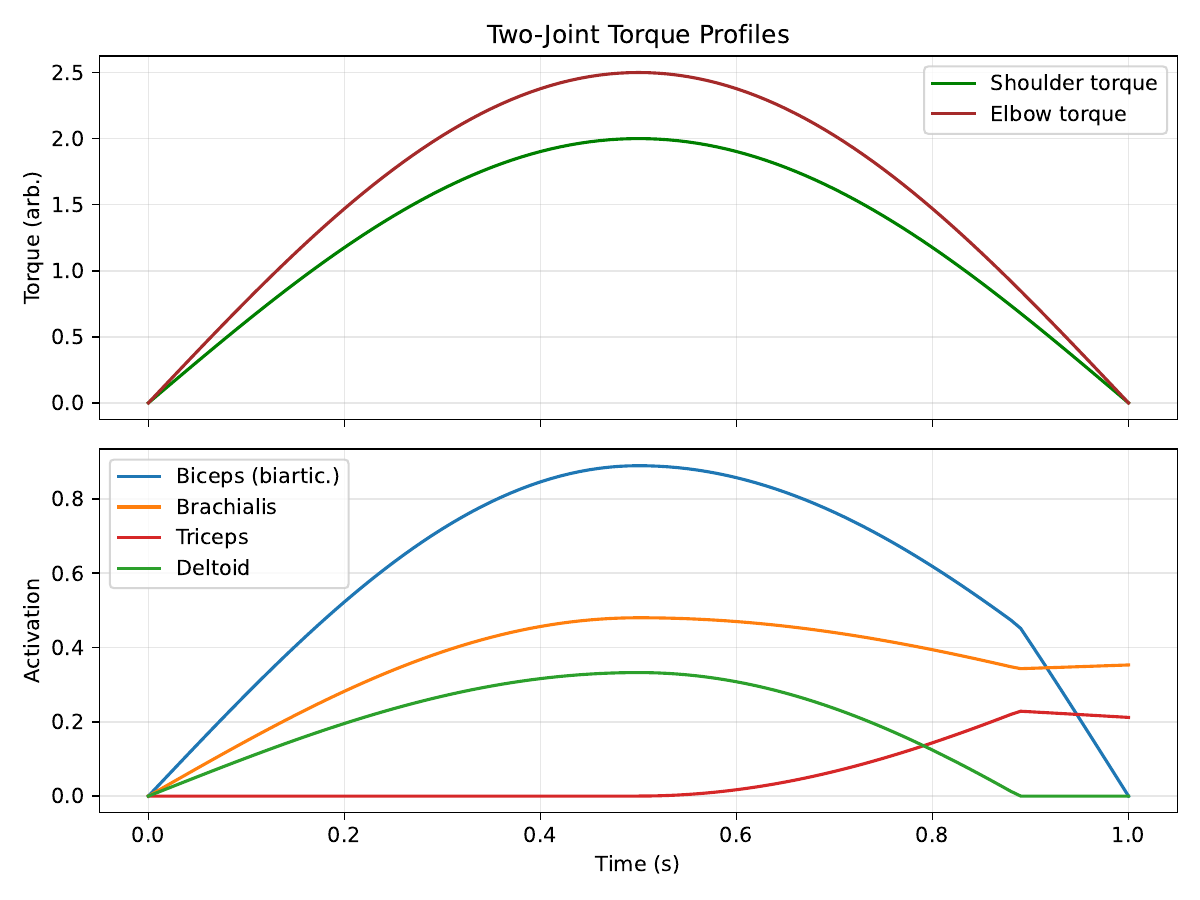}
    \caption{TFPD simulation of a two-joint (shoulder–elbow) system with a biarticular Biceps. The antagonist Triceps is recruited during torque offset despite the coupling introduced by the two-joint muscle, and the shoulder flexor (Deltoid) and Biceps co-activate smoothly. The change in slope at $t \approx 0.9$ corresponds to an active-set transition: Deltoid and Brachialis reach zero activation and exit the active set, redistributing the required torque among the remaining muscles (Biceps and Triceps).}
    \label{fig:twojoint}
\end{figure}

Several features are noteworthy. First, the biarticular Biceps is used to satisfy torque requirements at both joints simultaneously, and its activation profile smoothly transitions with the changing torques. Second, despite the more complex null-space geometry (dimension 2), the TFPD continues to activate the antagonist Triceps during the torque offset phase, exactly as predicted by the sufficient condition of Theorem~\ref{thm:active} applied to the elbow torque constraint. Third, the anterior deltoid, which acts only at the shoulder, shows a modulation pattern that closely follows the shoulder torque demand.

This example demonstrates that the TFPD framework requires no modification when moving from a single joint to multiple joints with biarticular muscles. The projection onto the convex polytope $\mathcal{F}_{\boldsymbol{\tau}}$ remains well-defined, the QP structure is preserved, and the active-set switching mechanism continues to generate antagonist co-contraction wherever the unconstrained update becomes infeasible. The coupling between joints introduced by biarticular muscles is naturally accommodated within the null space of the extended moment-arm matrix.

\section{Experimental Validation with Real EMG Data}
\label{sec:experimental}

To further assess the predictive power of the TFPD beyond idealized simulations, we tested the model on a publicly available dataset of elbow flexion–extension movements recorded from ten healthy subjects (6 males, 4 females) \cite{toro2023emg}. The dataset provides surface EMG from the biceps and triceps brachii together with joint angle estimated by inertial measurement units. Subjects performed flexion–extension cycles under three external loads; we analyzed the no-load ($0\,\text{g}$) and light-load ($1360\,\text{g}$, corresponding to $3\,\text{lb}$) conditions, as these produced sufficiently slow movements (mean cycle duration $\approx 11.3$–$12.4\,\text{s}$) for the gravitational component to dominate the net muscle torque.

For each subject, we computed the quasi-static gravitational torque about the elbow as
\begin{equation}
    \tau(t) = C_{\text{grav}} \, \sin\theta(t),
\end{equation}
where $\theta(t)$ is the instantaneous elbow angle and the constant $C_{\text{grav}}$ is determined from the subject's weight, forearm length, hand length, and the external load (see Appendix~\ref{app:experimental} for details). The TFPD was then driven by this torque profile, exactly as in the idealized simulations, without any parameter tuning.

\textbf{On the moment-arm values used for this comparison.} The moment-arm ratios $A=[2.0,1.5,-2.5]$ used throughout this paper (Section~\ref{sec:model}) were chosen as an illustrative, round-number configuration to demonstrate the projection mechanism, not fit to subject-specific anatomy; we reuse the same fixed ratios here rather than calibrating a separate value per subject. ``Without any parameter tuning'' therefore refers specifically to the \textit{TFPD update rule itself} (Eq.~\ref{eq:tfpd}, which contains no cost weights or free parameters to fit), not to the choice of moment-arm ratios, which is a modeling assumption whose anatomical realism we have not independently verified against subject-specific imaging or cadaveric moment-arm data. Because the reported correlations concern the \textit{shape} of the Triceps activation envelope over the movement cycle rather than its absolute amplitude, we expect the qualitative agreement to be only moderately sensitive to the exact ratio $r_3/(r_1+r_2)$; a systematic sensitivity analysis of the experimental correlation to this ratio, and validation against subject-specific or literature-derived anatomical moment arms, is an important direction for future work and is not yet included in the present study.

Figure~\ref{fig:experimental} shows a representative comparison for one subject (Subject~6, $1360\,\text{g}$ condition). The gravitational torque remains positive throughout the movement because the elbow angle varies between approximately $2.5^\circ$ and $130^\circ$, for which $\sin\theta > 0$. The model reproduces the overall envelope of the triceps activation with fidelity, capturing the gradual rise and fall of antagonist activity that accompanies the cyclic motion. Table~\ref{tab:experimental} reports the Pearson correlation coefficient and the normalized root-mean-square error between the simulated and experimental triceps envelopes for each subject, together with the across-subject mean and standard deviation. The TFPD achieves a mean Pearson correlation of $0.68$ (no load) and $0.73$ ($1360\,\text{g}$ load), with corresponding normalized RMSE values of $0.17$ and $0.16$.

\begin{figure}[htbp]
    \centering
    \includegraphics[width=\textwidth]{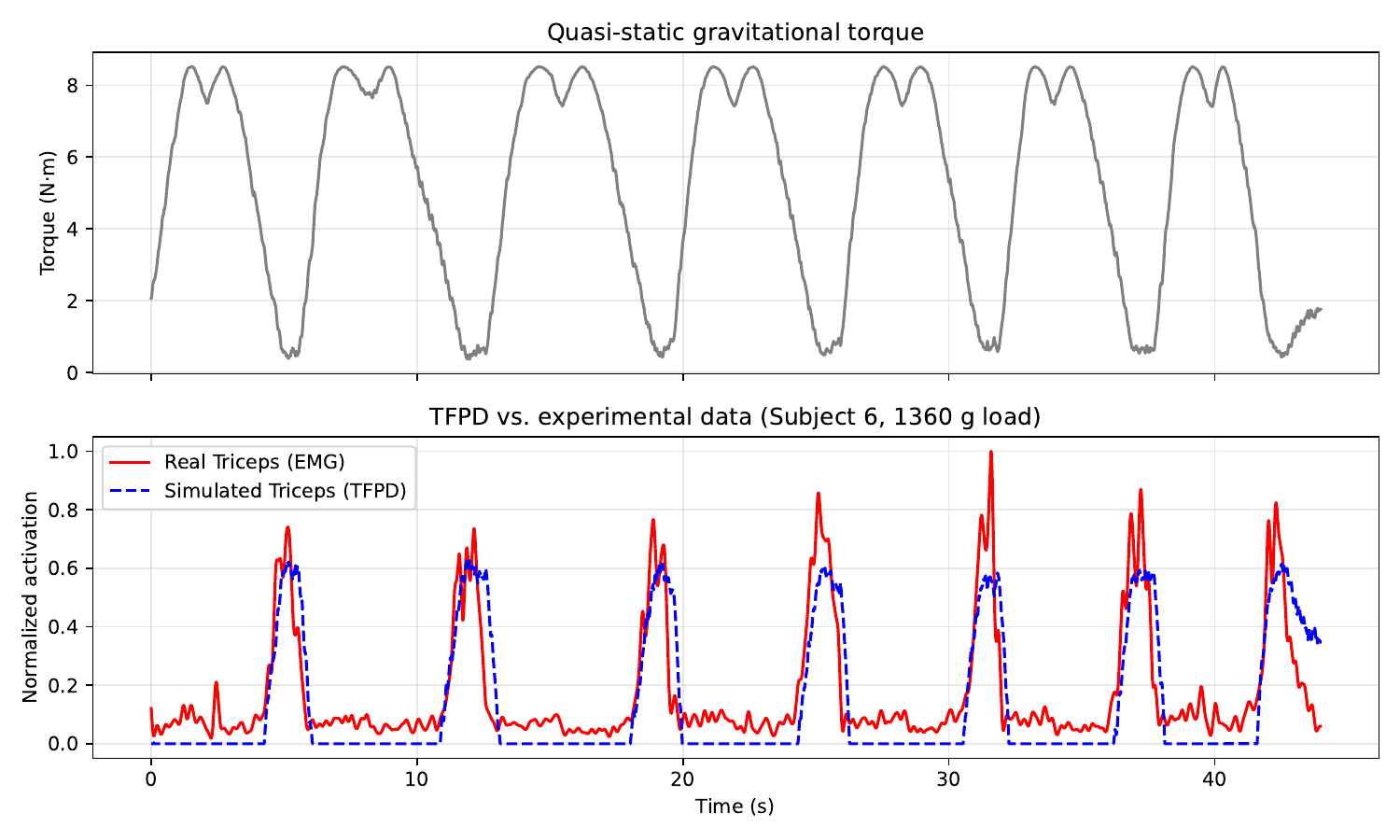}
    \caption{Experimental validation of the TFPD. Top: quasi-static gravitational torque $\tau(t)$ for one cycle of Subject~6 ($1360\,\text{g}$ load). The torque remains positive because $\sin\theta > 0$ throughout the movement. Bottom: normalized triceps EMG envelope (red) and TFPD simulation output (blue).}
    \label{fig:experimental}
\end{figure}

\begin{table}[htbp]
\centering
\caption{Quantitative comparison between TFPD predictions and experimental triceps EMG envelopes ($n=10$ subjects).}
\label{tab:experimental}
\begin{tabular}{@{}lccccc@{}}
\toprule
& \multicolumn{2}{c}{No load ($0\,\text{g}$)} & \multicolumn{2}{c}{Light load ($1360\,\text{g}$)} \\
\cmidrule(lr){2-3} \cmidrule(lr){4-5}
Subject & RMSE & Pearson $r$ & RMSE & Pearson $r$ \\
\midrule
1  & 0.20 & 0.66 & 0.13 & 0.71 \\
2  & 0.19 & 0.68 & 0.14 & 0.74 \\
3  & 0.16 & 0.69 & 0.15 & 0.74 \\
4  & 0.19 & 0.70 & 0.12 & 0.79 \\
5  & 0.18 & 0.54 & 0.20 & 0.66 \\
6  & 0.17 & 0.73 & 0.12 & 0.86 \\
7  & 0.13 & 0.68 & 0.15 & 0.74 \\
8  & 0.15 & 0.67 & 0.19 & 0.66 \\
9  & 0.17 & 0.71 & 0.17 & 0.69 \\
10 & 0.12 & 0.75 & 0.19 & 0.72 \\
\midrule
Mean $\pm$ SD & $0.166 \pm 0.025$ & $0.682 \pm 0.055$ & $0.157 \pm 0.028$ & $0.731 \pm 0.056$ \\
\bottomrule
\end{tabular}
\end{table}

These results provide an initial quantitative benchmark for the TFPD. The correlation values are moderate to high, particularly considering that the model contains no free parameters in its update rule and uses a deliberately simplified representation of both muscle mechanics and activation dynamics. The slightly higher correlation under light load is consistent with the expectation that a larger gravitational torque yields a clearer signal relative to unmodeled sources of variability (e.g., inertial effects, co-contraction for joint stabilization, or inter-cycle kinematic variability). A full dynamic validation incorporating directly measured joint torque, inertial contributions, and subject-specific moment arms is an important direction for future work.

\section{Discussion}
\label{sec:discussion}

We introduced the Torque Fiber Proximal Dynamics—a variational inequality framework for muscle redundancy. Theorem~\ref{thm:active} provides a precise sufficient condition for antagonist activation, while Theorem~\ref{thm:monotone} places the system within the class of maximal monotone operators and sweeping processes.

\textbf{Why the Euclidean metric?} The choice of the Euclidean projection is the simplest embodiment of the temporal smoothness principle: it selects the feasible activation that minimally deviates from the previous one. It is also the proximal operator associated with quadratic cost functionals; for more general Bregman divergences, one would obtain generalized projections, but the qualitative behavior remains similar under appropriate regularity conditions. This choice is thus natural and provides a canonical starting point for exploring projection-based control.

\textbf{Why does the TFPD produce co-contraction?} The projection onto a moving convex polytope with asymmetric constraints naturally induces boundary activation when the unconstrained update is infeasible. The sensitivity analysis demonstrates that this is a robust geometric phenomenon, not a model artifact: co-contraction persists under moment-arm scaling, addition of muscles, and torque amplitude variations. The two-joint extension further confirms that the mechanism is independent of the specific joint configuration and accommodates intersegmental coupling introduced by biarticular muscles. The experimental validation demonstrates that this mechanism is not restricted to idealized simulations: the same parameter-free model reproduces the envelope of triceps activation recorded from human subjects during cyclic flexion–extension tasks.

\textbf{Two regimes of validation.}
It is important to distinguish the two dynamical regimes examined in this paper.
The idealized simulations (Sections~6 and~7) employ a rapid sinusoidal torque profile
($1$\,s half-cycle), which produces the classic triphasic-like antagonist burst characteristic
of ballistic movements \cite{hallet1975}.
In contrast, the experimental validation (Section~\ref{sec:experimental}) uses
relatively slow cyclic movements (mean cycle duration $\approx 11$--$12$\,s) in which
the gravitational torque dominates and inertial contributions are negligible.
Under these quasi-static conditions the TFPD generates a smooth, sustained triceps
activation that follows the envelope of the recorded EMG, rather than a transient burst.
The fact that the same parameter-free model captures antagonist recruitment in both
fast-simulated and slow-experimental settings supports the generality of the temporal
consistency principle, but we emphasize that the experimental validation is a
surrogate for ballistic tests, not a replacement. Direct validation against fast,
ballistic movements with simultaneously recorded joint torque remains a central
objective for future work.

\textbf{Comparison with classical methods.} Table~\ref{tab:metrics} shows that the TFPD is the only model that achieves substantial Triceps co-contraction while maintaining a low fatigue index. Unlike MN, MS, MM, and WMN—which produce zero antagonist activation—the TFPD generates physiologically realistic co-contraction as a natural consequence of the temporal consistency principle. The random baseline (RND) achieves high Triceps sparsity (0.933) but at the cost of an excessively large fatigue index (126.3), confirming that random activation is not a viable physiological strategy. The TFPD uniquely combines moderate co-contraction with smooth, low-fatigue trajectories.

\textbf{Key insight:} Under the TFPD hypothesis, co-contraction is not a cost-minimizing strategy but a geometric consequence of constrained temporal evolution. Unlike classical formulations, no global objective is assumed. Instead, behavior emerges from a local projection rule equivalent to proximal point iterations. \textbf{Hallett et al. (1975) demonstrated that the triphasic pattern in fast flexion is ``centrally programmed.'' Our model suggests a geometric interpretation of this phenomenon: the antagonist burst may not require a separate neural command, but rather emerges as a constraint-induced redistribution when a temporal smoothness principle operates under physiological bounds. Whether this geometric mechanism underlies the central programming observed experimentally remains an open empirical question.}

\textbf{Normative nature of the model.} It is important to emphasize that the TFPD is a \textit{normative} model: it shows what would happen \textit{if} the motor system followed a temporal smoothness principle under physiological bounds. The fact that its predictions are qualitatively consistent with experimental EMG patterns—and that this consistency is robust under parameter perturbations and across different loading conditions—is suggestive, but it does not constitute proof that the nervous system literally implements this projection. Whether the actual neural circuitry approximates such proximal dynamics is an open empirical question.

\textbf{Relation to optimal control.}
The update rule $\mathbf{a}_t = \arg\min_{\mathbf{a}\in\mathcal{F}_{\tau(t)}}\|\mathbf{a}-\mathbf{a}_{t-1}\|^2$
can be viewed as the optimal solution of a one-step quadratic cost problem
in which the control cost (change in activation) dominates the state cost.
Formally, it is the limiting case of an optimal feedback control problem
with a running cost $c(\mathbf{a},\mathbf{u}) = \|\mathbf{u}\|^2 + \varepsilon \ell(\mathbf{a})$
as $\varepsilon \to 0$, where $\mathbf{u} = \mathbf{a} - \mathbf{a}_{t-1}$.
This provides a normative justification for the temporal smoothness principle
as the limit of a minimal-intervention strategy.

\textbf{Connection to Projected Dynamical Systems.}
The TFPD update rule $\mathbf{a}_t = \Pi_{\mathcal{F}_{\tau(t)}}(\mathbf{a}_{t-1})$ is an instance of a \textbf{projected dynamical system} (PDS) on the moving convex set $\mathcal{F}_{\tau(t)}$. PDSs are well-studied in variational inequality theory \cite{nagurney1996} and enjoy strong stability and uniqueness properties. This connection places the TFPD within a rigorous mathematical framework and opens the possibility of proving global convergence and Lyapunov stability for the muscle activation trajectories—subjects for future work.

\textbf{Connection to muscle synergies.}
The activation patterns generated by the TFPD qualitatively resemble the muscle synergies extracted from experimental EMG data (e.g., d'Avella et al., 2006). In both frameworks, high-dimensional muscle commands are structured by a low-dimensional latent space. The crucial difference is that synergies are typically viewed as \textit{structural primitives}—fixed building blocks that the nervous system combines to produce movement—whereas the TFPD generates structured patterns dynamically, without any pre-specified synergy library. Instead, the structure emerges purely from the geometry of the moving torque fiber and the temporal consistency principle. This observation suggests the intriguing possibility that some experimentally observed synergies may be epiphenomena of an underlying projection-based control strategy, rather than hard-wired neural modules. Exploring this connection quantitatively is a promising direction for future work.

\textbf{Positioning with respect to classical redundancy-resolution methods.}
Table~\ref{tab:comparison} provides a conceptual comparison between the TFPD and the two main families of classical approaches: instantaneous cost-minimization (e.g., minimum-norm, metabolic cost, muscle stress) and optimal feedback control. The distinguishing feature of the TFPD is its \textbf{tuning-free} nature—no physiological cost function, no penalty weights, and no co-contraction bounds are required. Despite this simplicity, antagonist co-contraction emerges as a structural consequence of the projection geometry, and the model predictions are quantitatively consistent with experimental EMG from ten subjects.

\begin{table}[htbp]
\centering
\caption{Conceptual comparison between the TFPD and classical redundancy-resolution methods.}
\label{tab:comparison}
\renewcommand{\arraystretch}{1.2}
\small
\begin{tabularx}{\textwidth}{@{}l *{3}{>{\centering\arraybackslash}X}@{}}
\toprule
Feature & Instantaneous cost minimization & Optimal feedback control & \textbf{TFPD (this work)} \\
\midrule
Physiological cost function required? & Yes & Yes (implicit) & \textbf{No} \\
Tunable weights / penalty parameters? \rule[-1.2ex]{0pt}{4.5ex} & Yes & Yes & \textbf{No} \\
Explicit co-contraction bounds? & Often required & Often required & \textbf{Not required} \\
Temporal continuity enforced? & No & Via state dynamics & \textbf{Via projection} \\
EMG-validated on human subjects? & Possible & Possible & \textbf{Yes (10 subjects)} \\
Multi-joint generalization shown? & Possible & Yes & \textbf{Yes (simulation)} \\
\bottomrule
\end{tabularx}
\end{table}

\textbf{Limitations and future work.} The model is deliberately simplified to isolate the geometric mechanism. Several physiological complexities are intentionally set aside: (i) the idealized constant moment arms neglect the state-dependent variations of muscle lever arms with joint angle, and, for the experimental validation in Section~\ref{sec:experimental}, the specific moment-arm ratios used were not calibrated to subject-specific anatomy (see the discussion in Section~\ref{sec:experimental}); (ii) the absence of force-length and force-velocity relationships means the model cannot capture the intrinsic nonlinear mechanics of muscle tissue; (iii) activation dynamics—the temporal filtering between neural excitation and mechanical force production—are not included; (iv) the torque profile is assumed exogenous, whereas in biological systems it arises from limb dynamics and environmental interactions; and (v) the experimental validation is restricted to relatively slow cyclic movements where a quasi-static torque approximation is justified. Extending the TFPD from an affine plane to a curved manifold that incorporates these physiological constraints, calibrating moment arms to subject-specific anatomy, and validating it against fast, ballistic movements with directly measured joint torques, are important directions for future work. Furthermore, establishing stronger theoretical properties—such as global convergence, Lipschitz continuity of the trajectory, or the existence of a Lyapunov function for the TFPD—and a primitive-level (rather than solution-conditional) characterization of when flexor saturation occurs in Theorem~\ref{thm:active}, remain open challenges that would elevate the framework from a descriptive model to a full-fledged theory of muscle coordination.

\section{Conclusion}
\label{sec:conclusion}

The Torque Fiber Proximal Dynamics provides a rigorous framework for muscle redundancy that is \textbf{free of tunable cost-weight parameters}. Antagonist co-contraction emerges as a consequence of active-set switching under the projection dynamics, formalized by a precise sufficient condition based on the KKT conditions. Comprehensive numerical validation, including comparison against five classical baseline models and a random baseline, sensitivity analysis, experimental validation against real EMG data from ten subjects, and a three-dimensional visualization of the activation trajectory, provide evidence for the robustness and generality of the phenomenon. The TFPD offers a geometrically grounded, normative account of antagonist recruitment and provides a fresh perspective on the classical problem of muscle redundancy, while also opening new theoretical and experimental avenues for future investigation.

\appendix
\section{Quasi-Static Torque Computation}
\label{app:experimental}

For the experimental validation (Section~\ref{sec:experimental}), the gravitational torque about the elbow was computed as
\[
\tau(t) = C_{\text{grav}} \, \sin\theta(t),
\]
where $C_{\text{grav}} = (m_{\text{fa}} \, r_{\text{com}} + m_{\text{load}} \, L_{\text{total}})\, g$.

The forearm-plus-hand mass was estimated as $m_{\text{fa}} = 0.022\,M$, with $M$ the body mass of each subject. The center-of-mass distance from the elbow was taken as $r_{\text{com}} = 0.682\,L_{\text{total}}$, where $L_{\text{total}} = L_{\text{forearm}} + L_{\text{hand}}$ is the sum of forearm and hand lengths reported in the dataset. The load mass $m_{\text{load}}$ was $0\,\text{kg}$ and $1.36\,\text{kg}$ for the no-load and $1360\,\text{g}$ conditions, respectively. Gravitational acceleration was $g = 9.81\,\text{m/s}^2$, and the joint angle $\theta$ was converted from degrees to radians before computing the sine.

\section{Python Code}
\label{app:code}

The following code generates all figures used in the paper: the single-joint model comparison (Fig.~\ref{fig:comparison}), the four-muscle sensitivity (Fig.~\ref{fig:sensitivity_muscles}), the torque amplitude variation (Fig.~\ref{fig:sensitivity_torque}), the two-joint biarticular muscle simulation (Fig.~\ref{fig:twojoint}), the experimental validation plot (Fig.~\ref{fig:experimental}), and the three-dimensional trajectory visualization (Fig.~\ref{fig:3d_trajectory}). It uses the SLSQP solver via \texttt{scipy.optimize.minimize} and exports all plots as PDF files. The random (RND) baseline used in Table~\ref{tab:metrics} is included below (\texttt{simulate\_random}); all other baselines (MN, MS, MM, WMN) follow the same \texttt{solve\_qp} pattern with the cost function in Section~\ref{sec:benchmarks} substituted in. The experimental-data path in \texttt{Fig 5} is a placeholder relative path (\texttt{./data/}); readers should point \texttt{DATA\_ROOT} at their local copy of the dataset of \cite{toro2023emg}.

\begin{verbatim}
import numpy as np
import matplotlib.pyplot as plt
from scipy.optimize import minimize
from scipy.signal import butter, filtfilt
from mpl_toolkits.mplot3d.art3d import Poly3DCollection
from scipy.spatial import ConvexHull
import os

# ---------- QP Solver ----------
def solve_qp(a_prev, A, tau):
    n = len(a_prev)
    fun = lambda a: 0.5 * np.sum((a - a_prev)**2)
    cons = [{'type': 'eq', 'fun': lambda a: A @ a - tau}]
    bounds = [(0, 1)] * n
    res = minimize(fun, a_prev, bounds=bounds, constraints=cons,
                   method='SLSQP', options={'maxiter': 100, 'ftol': 1e-10})
    return res.x

def simulate(A, tau):
    if tau.ndim == 1:
        tau = tau.reshape(1, -1)
    n_steps = tau.shape[1]
    n_mus = A.shape[1]
    a_traj = np.zeros((n_steps, n_mus))
    a_prev = np.zeros(n_mus)
    for i in range(n_steps):
        tt = tau[:, i]
        a_traj[i] = solve_qp(a_prev, A, tt)
        a_prev = a_traj[i]
    return a_traj

# ---------- Random (RND) baseline used in Table 1 ----------
def simulate_random(A, tau, rng=None):
    """At each time step, draw a uniformly random feasible point on the
    torque fiber F_tau = {a in [0,1]^n : A a = tau} via rejection/hit-and-run
    on the affine slice. Uses a simple projected-random-walk sampler:
    draw a random direction in the null space of A, take a bounded random
    step, and project back onto the affine constraint A a = tau."""
    if rng is None:
        rng = np.random.default_rng(0)
    if tau.ndim == 1:
        tau = tau.reshape(1, -1)
    n_steps = tau.shape[1]
    n_mus = A.shape[1]
    # any feasible particular solution for the first time step
    a0 = np.linalg.lstsq(A, tau[:, 0], rcond=None)[0]
    a0 = np.clip(a0, 0, 1)
    a_traj = np.zeros((n_steps, n_mus))
    a_prev = a0
    # orthonormal basis of the null space of A (constant across steps here
    # since A is fixed; recompute per-step if A varies with time)
    _, _, Vt = np.linalg.svd(A)
    null_basis = Vt[A.shape[0]:].T  # columns span null(A)
    for i in range(n_steps):
        tt = tau[:, i]
        # start from any feasible particular solution for this tau
        a_particular = np.linalg.lstsq(A, tt, rcond=None)[0]
        # random combination within the null space, then clip+reproject
        for _ in range(200):  # rejection loop for a point in [0,1]^n
            coeffs = rng.normal(size=null_basis.shape[1])
            candidate = a_particular + null_basis @ coeffs
            if np.all(candidate >= 0) and np.all(candidate <= 1):
                a_prev = candidate
                break
        else:
            # fall back to the least-squares particular solution, clipped
            a_prev = np.clip(a_particular, 0, 1)
        a_traj[i] = a_prev
    return a_traj

# ---------- Settings ----------
t = np.linspace(0, 1, 101)
tau_max = 2.5
tau_base = tau_max * np.sin(np.pi * t)
muscle_names_3 = ['Biceps', 'Brachialis', 'Triceps']
colors_3 = ['#1f77b4', '#ff7f0e', '#d62728']

# ---------- Fig 1: MN vs TFPD ----------
A1 = np.array([[2.0, 1.5, -2.5]])
a_tfpd = simulate(A1, tau_base)
a_mn = np.zeros_like(a_tfpd)
for i, tt in enumerate(tau_base):
    a_mn[i] = solve_qp(np.zeros(3), A1, tt)

fig, axes = plt.subplots(1, 2, figsize=(10, 4), sharex=True)
for j, (name, c) in enumerate(zip(muscle_names_3, colors_3)):
    axes[0].plot(t, a_mn[:, j], color=c, label=name, linewidth=1.5)
axes[0].set_title('(a) Minimum-norm'); axes[0].legend(); axes[0].set_ylim(-0.05, 1.05)
for j, (name, c) in enumerate(zip(muscle_names_3, colors_3)):
    axes[1].plot(t, a_tfpd[:, j], color=c, label=name, linewidth=1.5)
axes[1].set_title('(b) TFPD'); axes[1].legend(); axes[1].set_ylim(-0.05, 1.05)
plt.tight_layout(); plt.savefig('elbow_comparison.pdf'); plt.close()

# ---------- Fig 2: Four-Muscle ----------
A4 = np.array([[2.0, 1.5, -2.5, 1.0]])
a4 = simulate(A4, tau_base)
muscle_names_4 = ['Biceps', 'Brachialis', 'Triceps', 'Brachioradialis']
colors_4 = ['#1f77b4', '#ff7f0e', '#d62728', '#2ca02c']
plt.figure(figsize=(8, 4))
for j, (name, c) in enumerate(zip(muscle_names_4, colors_4)):
    plt.plot(t, a4[:, j], color=c, label=name, linewidth=1.5)
plt.legend(); plt.xlabel('Time (s)'); plt.ylabel('Activation')
plt.title('Four-Muscle Model'); plt.grid(alpha=0.3)
plt.tight_layout(); plt.savefig('sensitivity_four_muscles.pdf'); plt.close()

# ---------- Fig 3: Torque Amplitude ----------
amplitudes = [1.5, 2.5, 3.5]
linestyles = ['--', '-', '-.']
plt.figure(figsize=(7, 4))
for amp, ls in zip(amplitudes, linestyles):
    tau_amp = amp * np.sin(np.pi * t)
    a_amp = simulate(A1, tau_amp)
    plt.plot(t, a_amp[:, 2], linestyle=ls, linewidth=1.8,
             label=f'$\\tau_{{\\max}} = {amp}$')
plt.xlabel('Time (s)'); plt.ylabel('Triceps Activation')
plt.title('Sensitivity to Torque Amplitude')
plt.legend(); plt.grid(alpha=0.3)
plt.tight_layout(); plt.savefig('sensitivity_torque_amplitude.pdf'); plt.close()

# ---------- Fig 4: Two-Joint ----------
A2 = np.array([[1.5, 0.0,  0.0, 2.0],
               [2.0, 1.5, -2.5, 0.0]])
tau_sh = 2.0 * np.sin(np.pi * t)
tau_el = tau_base
tau_2d = np.vstack([tau_sh, tau_el])
a2 = simulate(A2, tau_2d)
muscle_names_2 = ['Biceps (biartic.)', 'Brachialis', 'Triceps', 'Deltoid']
colors_2 = ['#1f77b4', '#ff7f0e', '#d62728', '#2ca02c']
fig, axes = plt.subplots(2, 1, figsize=(8, 6), sharex=True)
axes[0].plot(t, tau_sh, label='Shoulder torque', color='green', linewidth=1.5)
axes[0].plot(t, tau_el, label='Elbow torque', color='brown', linewidth=1.5)
axes[0].set_ylabel('Torque (arb.)'); axes[0].legend(); axes[0].grid(alpha=0.3)
axes[0].set_title('Two-Joint Torque Profiles')
for j, (name, c) in enumerate(zip(muscle_names_2, colors_2)):
    axes[1].plot(t, a2[:, j], color=c, label=name, linewidth=1.5)
axes[1].set_xlabel('Time (s)'); axes[1].set_ylabel('Activation')
axes[1].legend(); axes[1].grid(alpha=0.3)
plt.tight_layout(); plt.savefig('twojoint_biarticular.pdf'); plt.close()

# ---------- Fig 5: Experimental ----------
# NOTE: replace DATA_ROOT with the local path to the dataset of [toro2023emg]
# (https://doi.org/10.5281/zenodo.7946782). A relative path is used here so
# the script is portable across machines/readers.
DATA_ROOT = os.environ.get("TFPD_DATA_ROOT", "./data")
subject_dir = os.path.join(DATA_ROOT, "6")
info_file = os.path.join(subject_dir, "subject_info.txt")
if os.path.exists(info_file):
    info = {}
    with open(info_file, 'r') as f:
        for line in f:
            if ':' in line:
                key, val = line.strip().split(':')
                info[key.strip()] = val.strip()
    weight = float(info['weight']); arm_len = float(info['arm_length']); hand_len = float(info['hand_length'])
    g_val = 9.81; L_total = arm_len + hand_len
    m_fa = 0.022 * weight; r_com = 0.682 * L_total; m_load = 1.36
    C_grav = (m_fa * r_com + m_load * L_total) * g_val
    data_file = os.path.join(subject_dir, "6_flex_test_1360.txt")
    if os.path.exists(data_file):
        data = np.loadtxt(data_file, delimiter='\t')
        emg_triceps = data[:, 3]; angle_deg = data[:, 4]; fs = 512.0
        def butter_lowpass(cutoff, fs, order=4):
            nyq = 0.5 * fs; normal_cutoff = cutoff / nyq
            b, a = butter(order, normal_cutoff, btype='low', analog=False)
            return b, a
        rectified = np.abs(emg_triceps); b, a = butter_lowpass(4.0, fs, order=4)
        env_real = filtfilt(b, a, rectified); env_real = env_real / np.max(env_real)
        angle_rad = np.deg2rad(angle_deg); torque = C_grav * np.sin(angle_rad)
        A_exp = np.array([[2.0, 1.5, -2.5]]); a_sim = simulate(A_exp, torque); env_sim = a_sim[:, 2]
        time = np.arange(len(env_real)) / fs
        fig, axes = plt.subplots(2, 1, figsize=(10, 6), sharex=True)
        axes[0].plot(time, torque, color='gray', linewidth=1.5)
        axes[0].set_ylabel('Torque (N·m)'); axes[0].set_title('Quasi-static gravitational torque'); axes[0].grid(alpha=0.3)
        axes[1].plot(time, env_real, label='Real Triceps (EMG)', color='red', linewidth=1.5)
        axes[1].plot(time, env_sim, label='Simulated Triceps (TFPD)', color='blue', linestyle='--', linewidth=1.5)
        axes[1].set_xlabel('Time (s)'); axes[1].set_ylabel('Normalized activation')
        axes[1].set_title('TFPD vs. experimental data (Subject 6, 1360 g load)')
        axes[1].legend(); axes[1].grid(alpha=0.3)
        plt.tight_layout(); plt.savefig('comparison_static_torque.pdf'); plt.close()
        print("Figure 5 saved: comparison_static_torque.pdf")

# ---------- Fig 6: 3D Trajectory ----------
plt.rcParams.update({"font.family": "serif", "font.serif": ["Times New Roman", "DejaVu Serif"], "font.size": 10, "axes.labelsize": 11, "legend.fontsize": 9, "figure.titlesize": 12, "mathtext.fontset": "stix"})
A_traj = simulate(A1, tau_base)
fig = plt.figure(figsize=(8, 7), dpi=300); ax = fig.add_subplot(111, projection='3d')
ax.set_facecolor('white'); ax.view_init(elev=25, azim=-45)
corners = np.array([[0,0,0],[1,0,0],[1,1,0],[0,1,0],[0,0,1],[1,0,1],[1,1,1],[0,1,1]])
edges = [(0,1),(1,2),(2,3),(3,0),(4,5),(5,6),(6,7),(7,4),(0,4),(1,5),(2,6),(3,7)]
for e in edges: ax.plot(corners[e,0], corners[e,1], corners[e,2], color='gray', linewidth=0.6, alpha=0.5)
def fiber_polygon(tau_val):
    pts = []
    for x in [0,1]: 
        for y in [0,1]:
            z = (2*x + 1.5*y - tau_val)/2.5
            if 0 <= z <= 1: pts.append([x,y,z])
    for x in [0,1]:
        for z in [0,1]:
            y = (tau_val + 2.5*z - 2*x)/1.5
            if 0 <= y <= 1: pts.append([x,y,z])
    for y in [0,1]:
        for z in [0,1]:
            x = (tau_val + 2.5*z - 1.5*y)/2.0
            if 0 <= x <= 1: pts.append([x,y,z])
    pts = np.unique(np.array(pts), axis=0)
    if len(pts) < 3: return None
    N = np.array([2.0,1.5,-2.5]); N = N/np.linalg.norm(N)
    U = np.array([1,-2/1.5,0]); U = U/np.linalg.norm(U); V = np.cross(N,U)
    proj = np.column_stack([np.dot(pts,U), np.dot(pts,V)])
    hull = ConvexHull(proj)
    return pts[hull.vertices]
taus_fibers = [0.0, 1.2, 2.0, 2.5]; colors_fiber = ['#882255','#44AA99','#DDCC77','#CC6677']
labels_fiber = [r'$\tau = 0.0$', r'$\tau = 1.2$', r'$\tau = 2.0$', r'$\tau = 2.5$']
for tv, cl, lb in zip(taus_fibers, colors_fiber, labels_fiber):
    poly = fiber_polygon(tv)
    if poly is not None:
        ax.add_collection3d(Poly3DCollection([poly], alpha=0.3, facecolor=cl, edgecolor=cl, linewidth=0.8))
        ctr = np.mean(poly, axis=0)
        ax.text(ctr[0], ctr[1], ctr[2]+0.03, lb, color='black', fontsize=8, ha='center')
step = 2
ax.plot(A_traj[::step,0], A_traj[::step,1], A_traj[::step,2], color='#117733', linewidth=2.0, label='Activation trajectory')
ax.scatter(A_traj[::step,0], A_traj[::step,1], A_traj[::step,2], color='#117733', s=10, alpha=0.7)
ax.scatter(*A_traj[0], color='green', s=70, edgecolors='black', label='Start')
ax.scatter(*A_traj[-1], color='red', s=70, edgecolors='black', label='End')
ax.text(0.65, 0.55, 0.08, 'Active-set boundary\n($a_3 = 0$)', color='#8B0000', fontsize=9, fontweight='bold',
        bbox=dict(boxstyle='round,pad=0.3', facecolor='white', alpha=0.85, edgecolor='#8B0000', linewidth=0.5), zorder=20)
ax.set_xlabel('Biceps ($a_1$)'); ax.set_ylabel('Brachialis ($a_2$)'); ax.set_zlabel('Triceps ($a_3$)')
ax.set_xlim(0,1); ax.set_ylim(0,1); ax.set_zlim(0,1)
ax.set_xticks([0,0.5,1]); ax.set_yticks([0,0.5,1]); ax.set_zticks([0,0.5,1])
ax.xaxis.pane.fill = False; ax.yaxis.pane.fill = False; ax.zaxis.pane.fill = False
ax.xaxis.pane.set_edgecolor('white'); ax.yaxis.pane.set_edgecolor('white'); ax.zaxis.pane.set_edgecolor('white')
ax.grid(True, linestyle=':', alpha=0.4); ax.legend(loc='upper left', framealpha=0.9)
plt.title('Torque Fiber Proximal Dynamics (TFPD)', pad=15)
plt.tight_layout(); plt.savefig('activation_trajectory_3d.pdf'); plt.close()

print("All 6 figures generated successfully.")
\end{verbatim}

\end{document}